\documentclass[letterpaper, 10 pt, conference]{ieeeconf}  

\IEEEoverridecommandlockouts        

\usepackage{mathtools}
\usepackage{setspace}

\usepackage{amsthm}

\theoremstyle{plain}
\newtheorem{theorem}{Theorem}

\newtheorem{lemma}{Lemma}
\newtheorem{assumption}{Assumption}

\theoremstyle{definition}
\newtheorem{definition}{Definition}

\newtheorem{example}{Example}
\theoremstyle{plain}

\theoremstyle{remark}
\newtheorem{remark}{Remark}

\newcommand{\mE}{\ensuremath{\mathbb{E}}}

\newcommand\gC{{\mathcal{C}}}

\providecommand{\keywords}[1]{\vspace*{0.5\baselineskip}\par\noindent\textbf{Keywords:} #1}

{%
\begin{list}{#1}{
\vspace{-\topsep}
\vspace{-\partopsep}
\setlength{\itemindent}{0cm}
\setlength{\rightmargin}{0cm}
\setlength{\listparindent}{0cm}
\settowidth{\labelwidth}{#1}
\setlength{\leftmargin}{\labelwidth}
\addtolength{\leftmargin}{\labelsep}
\setlength{\itemsep}{0cm}
}%
}%
{%
\end{list}
\vspace{-\topsep}
\vspace{-\partopsep}
}

{\begin{enumerate}%
\renewcommand{\theenumi}{\roman{enumi}}%
}%
{\end{enumerate}}

\usepackage{bbold}

\usepackage[sort, nocompress]{cite}
\usepackage{amsmath,amssymb,amsthm, amsfonts, mathrsfs}
\usepackage{graphicx}
\usepackage{dsfont}
\usepackage{xcolor}
\usepackage{comment}

\usepackage{cleveref}
\Crefname{equation}{}{}
\crefname{equation}{}{}
\usepackage{footmisc}
\usepackage{csquotes}
\let\labelindent\relax
\usepackage{enumitem}
\usepackage{bbm}
\usepackage{algorithm}
\usepackage{algpseudocode}
\usepackage{tabularx}
\usepackage{booktabs}

\usepackage{caption, subcaption}

\def\BibTeX{{\rm B\kern-.05em{\sc i\kern-.025em b}\kern-.08em
    T\kern-.1667em\lower.7ex\hbox{E}\kern-.125emX}}

\newcommand{\eg}{\textit{e.g., }}
    
\newcommand{\R}{\mathbb R}

\newcommand{\spy}{F} 
\newcommand{\rpy}{G} 
\newcommand{\spx}{Q} 

\title{\LARGE \bf
Fractional Risk Analysis of Stochastic Systems with Jumps and Memory
}

\author{Yimeng Sun$^{1*}$, Zhuoyuan Wang$^{1*}$, Xiaole Zhang$^{2*}$, Heng Ping$^{2}$, Jintang Xue$^{2}$, Paul Bogdan$^{2}$, Yorie Nakahira$^{1\dagger}$
  \thanks{$^{1}$Carnegie Mellon University. $^{2}$University of Southern California.}
  \thanks{$^{*}$Equal contribution. $^{\dagger}$Correspondence: {\tt\small yorie@cmu.edu}.}
}

\begin{document}

\maketitle
\thispagestyle{empty}
\pagestyle{empty}

\begin{abstract}

Accurate risk assessment is essential for safety-critical autonomous and control systems under uncertainty. In many real-world settings, stochastic dynamics exhibit asymmetric jumps and long-range memory, making long-term risk probabilities difficult to estimate across varying system dynamics, initial conditions, and time horizons. Existing sampling-based methods are computationally expensive due to repeated long-horizon simulations to capture rare events, while existing partial differential equation (PDE)-based formulations are largely limited to Gaussian or symmetric jump dynamics and typically treat memory effects in isolation.
In this paper, we address these challenges by deriving a space- and time-fractional PDE that characterizes long-term safety and recovery probabilities for stochastic systems with both asymmetric Lévy jumps and memory. This unified formulation captures nonlocal spatial effects and temporal memory within a single framework and enables the joint evaluation of risk across initial states and horizons. We show that the proposed PDE accurately characterizes long-term risk and reveals behaviors that differ fundamentally from systems without jumps or memory and from standard non-fractional PDEs. Building on this characterization, we further demonstrate how physics-informed learning can efficiently solve the fractional PDEs, enabling accurate risk prediction across diverse configurations and strong generalization to out-of-distribution dynamics.
\end{abstract}



\section{Introduction}

Safety is critical for autonomous and control systems operating under uncertainty. In many real-world settings, stochastic dynamics exhibit non-Gaussian fluctuations, asymmetric jumps, and long-range temporal dependence, as observed in applications including power grids~\cite{sia2018pmu} and physiological systems~\cite{yin2023fractional}. Under such conditions, safety must be characterized probabilistically. In particular, estimating the long-term probability that system trajectories enter unsafe regions is important for risk assessment and safety-critical control, but remains challenging due to the underlying stochastic complexity and extended time horizons involved~\cite{wang2022myopically, wang2024myopically}.

A variety of methods have been developed to estimate such long-term risk probabilities. Sampling-based approaches such as Monte Carlo simulation and its variants, including importance sampling and subset simulation, are widely used~\cite{rubino2009rare,cerou2012sequential,zuev2015subset}. However, these methods are often computationally expensive, since they must repeatedly simulate the stochastic dynamics over long horizons to capture rare failure events, and the calculations typically need to be redone for different system dynamics and initial conditions. Recently, partial differential equation (PDE)-based methods have shown promise by characterizing risk probabilities through deterministic equations derived from the underlying dynamics~\cite{chern2021safe,wang2025generalizable}. Compared with trajectory-based sampling, this perspective can jointly solve for risk probabilities across different initial states and time horizons, while also enabling physics-informed learning methods to serve as neural surrogate solvers for parametric and functional variants of the problem~\cite{raissi2019physics,pang2019fpinns,li2024physics}. However, existing PDE-based formulations remain largely limited to Gaussian noise and Markovian settings~\cite{chern2021safe, wang2025generalizable}, and only a few results have considered symmetric jumps or time-change effects in isolation, rather than stochastic systems that combine asymmetric jumps with memory~\cite{gao2014mean, yang2026efficient, abundo2020fractional}.

In this work, we extend the PDE-based risk characterization framework to a broader class of stochastic systems with both asymmetric jumps and memory. Specifically, we derive a space- and time-fractional PDE that characterizes the evolution of long-term risk probabilities for systems driven by general asymmetric Lévy processes and time-change effects (Theorem~\ref{thm:safe_prob_pde}-\ref{thm:rec_prob_pde}). We show that this fractional PDE accurately characterizes long-term risk values, capturing behaviors that differ fundamentally from systems without jumps or memory and from those described by standard non-fractional PDE formulations (Fig.~\ref{fig:rec_space_pde_vs_mc}-\ref{fig:rec_time_pde_vs_mc}). Building on this characterization, we further demonstrate how physics-informed learning techniques can be leveraged to efficiently estimate the corresponding fractional PDE solutions (Section~\ref{sec:safe_prob_exp}). The resulting framework integrates physical constraints with limited data and enables accurate prediction of long-term risk of diverse configurations while generalizing to out-of-distribution system dynamics (Fig.~\ref{fig:2D_ood_results} and Table~\ref{tab:pino_mc_comparison}).


\section{Problem Formulation}

\subsection{System Dynamics}

We consider a class of nonlinear stochastic control systems subject to heavy-tailed disturbances and time-change memory effects. We first introduce the underlying system without memory. Let $X_s \in \mathbb{R}^n$ denote the system state at time $s \ge 0$, evolving according to the stochastic differential equation
\begin{equation}\label{eq:x_trajectory}
    dX_s = f(X_s)\,ds + \sigma(X_s)\, dL_s^{(\alpha,\Lambda)},
\end{equation}
where $f:\mathbb{R}^n\to\mathbb{R}^n$ is the drift vector field, $\sigma:\mathbb{R}^n\to\mathbb{R}^{n\times k}$ is the diffusion coefficient, and $L_s^{(\alpha,\Lambda)}$ is an $\alpha$-stable Lévy process in $\mathbb{R}^n$ with $0<\alpha< 2$. The finite spectral measure $\Lambda$ on the unit sphere $S^{n-1}$ characterizes the directional distribution and skewness of the jump noise. In particular, the process is symmetric if and only if $\Lambda(A)=\Lambda(-A)$ for all measurable sets $A\subset S^{n-1}$.

\begin{assumption}
\label{asp:dynamics_regularity}
We assume that $f$ and $\sigma$ are bounded and globally Lipschitz.
\end{assumption}
Under Assumption~\ref{asp:dynamics_regularity}, there exists a unique strong solution to~\eqref{eq:x_trajectory} for the base system~\cite{applebaum2009levy}.
Note that dynamics~\eqref{eq:x_trajectory} can capture not only autonomous systems, but also closed-loop systems with memoryless state feedback control (\eg PID and robust controllers).


To model time-delay effects and memory phenomena, we introduce a stochastic time change. Let $(D_u)_{u \ge 0}$ be an $\beta$-stable subordinator with $0<\beta<1$, independent of $(X_s)_{s\ge0}$. Define its inverse (first-passage-time) process by
\begin{equation}
    E_t := \inf \{ u \ge 0 : D_u > t \}, \qquad t \ge 0.
\end{equation}
The state process with memory is defined via the following time change
\begin{equation}
\label{eq:time_change}
    Y_t := X_{E_t}, \qquad t \ge 0.
\end{equation}
This construction induces non-Markovian dynamics and captures heavy-tailed waiting-time effects.

\subsection{Safety Specification}

We formalize safety via a barrier function characterization.

\begin{definition}[Safe Set~\cite{ames2019control}]
Let $\phi : \mathbb{R}^n \to \mathbb{R}$ be a twice continuously differentiable function. The \emph{safe set} is defined as the superlevel set
\begin{equation}
    \mathcal{C} := \{ x \in \mathbb{R}^n : \phi(x) \ge 0 \}.
\end{equation}
The interior, boundary, and unsafe region are defined as
\begin{align}
\operatorname{int} (\gC) &= \{ x \in \mathbb{R}^n : \phi(x) > 0 \}, \\
    \partial \mathcal{C} &= \{ x \in \mathbb{R}^n : \phi(x) = 0 \}, \\
    \mathcal{C}^c &= \{ x \in \mathbb{R}^n : \phi(x) < 0 \}.
\end{align}
\end{definition}

\begin{assumption}
\label{asp:regular_safe_set}
    We assume all points on $\partial \gC$ are regular for the exterior $\gC^c$, meaning that the base process $X_t$ starting at any $x \in \partial \gC$ enters $\gC^c$ immediately with probability one.
\end{assumption}

\subsection{Long-Term Probabilities of Interest}

We study the long term safety and recovery probabilities of the delayed process $Y$, depending on whether the initial state $x$ lies inside the safe set $\gC$. If the system starts inside $\gC$, the safety probability is defined as 
\begin{equation}
\label{eq:safety_probability}
\spy(x,T):=\mathbb{P}\!\left(Y_t \in \gC,\ \forall t \in [0,T] \mid Y_0 = x \in \gC\right).
\end{equation}
If the system starts outside $\gC$, the recovery probability is defined as
\begin{equation}
\label{eq:recovery_probability}
\rpy(x,T):=\mathbb{P}\!\left(\exists t \in [0,T],\ Y_t \in \gC \mid Y_0 = x \notin \gC\right).
\end{equation}
To characterize these probabilities, define the first exit time from the
safe set and the first recovery time to the safe set 
\begin{equation}
\label{eq:exit_time}
\tau_{\gC} = \inf \{ t \ge 0 \mid Y_t \notin \gC \},
\end{equation}
\begin{equation}
\label{eq:recovery_time}
\rho_{\gC} = \inf \{ t \ge 0 \mid Y_t \in \gC \}.
\end{equation}
The safety and recovery probabilities can then be written as
\begin{equation}
\spy(x,T) = \mathbb{P}(\tau_{\gC} \geq T \mid Y_0 = x \in \gC),
\end{equation}
\begin{equation}
\rpy(x,T) = \mathbb{P}(\rho_{\gC} \leq T \mid Y_0 = x \notin \gC).
\end{equation}
Such long-term probabilities are essential for designing safe control methods to guarantee long-term safety of the system~\cite{wang2022myopically, wang2024myopically}, and are non-trivial to estimate because the computation scales with the horizon $T$ of interest as well as the number of initial states $x$ of interest~\cite{wang2025generalizable}. The research question for this paper is as follows.
\begin{itemize}
    \item When the stochastic system dynamics~\eqref{eq:x_trajectory} with jumps and time change~\eqref{eq:time_change} are given, what is the \textit{exact PDE characterization} of~\eqref{eq:safety_probability} and~\eqref{eq:recovery_probability} for different initial states $x$ and time horizons $T$?

    \item With the PDE characterization, how do we efficiently calculate their solutions to obtain the long-term probabilities~\eqref{eq:safety_probability} and~\eqref{eq:recovery_probability} of diverse system dynamics?
\end{itemize}

\section{Main Results}


In this section, we provide an analytical characterization of safety-critical events for stochastic systems with jumps and memory. Specifically, we derive fractional PDEs that directly characterize the long-term safety and recovery probabilities~\eqref{eq:safety_probability} and~\eqref{eq:recovery_probability}. This formulation provides a principled foundation for efficient computational methods, including numerical solvers and physics-informed learning approaches, as demonstrated in Section~\ref{sec:experiments}. Throughout the paper, we use upper-case letters (e.g., $X$) to denote random variables and lower-case letters (e.g., $x$) to denote their realizations.  

\begin{definition}[Infinitesimal Generator]
The infinitesimal generator $\mathcal A$ of
a stochastic process $\{ X_t  \in \R^n \}_{t \in \R_+}$ is
\begin{equation}
\label{eq:InfinitesimalGenerator}
\mathcal A\varphi(x) = \lim _{h\rightarrow 0} \frac{\mE_x\left[\varphi(X_{h})\right]-\varphi(x)}{h},
\end{equation}
whose domain $D(\mathcal{A})$ is the set of all functions $\varphi: \R^n \rightarrow \R$ such that the limit of \eqref{eq:InfinitesimalGenerator} exists for all $x \in \R^n$. 
\end{definition}
The infinitesimal generator of~\eqref{eq:x_trajectory} is then given by
\begin{equation}
\label{eq:space_fractional_generator}
    \mathcal{L}\varphi(x):=f(x)\cdot\nabla \varphi(x)+\sigma^{\alpha}(x)\,\mathcal{J}_{\alpha}\varphi(x),
\end{equation}
where $0<\alpha\leq2$ and the space-fractional jump operator $\mathcal{J}_{\alpha}$ is defined by its Fourier transform
\begin{equation}
    \mathcal{F}\{ \mathcal{J}_{\alpha}\varphi\}(\xi) = -\psi_\alpha(\xi)\,\mathcal{F}\{\varphi\}(\xi),
\end{equation} 
with characteristic exponent
\begin{equation}
    \psi_\alpha(\xi) = \int_{S^{n-1}}|\xi\cdot \theta|^{\alpha}\Big(1-i\,\tan\!\big(\tfrac{\pi\alpha}{2}\big)\,\mathrm{sign}(\xi\cdot \theta)\Big)\,\Lambda(d\theta).
\end{equation}
for $\alpha \in (0,2)\setminus\{1\}$, and
\begin{equation}
\psi_1(\xi)
=
\int_{S^{n-1}}
|\xi\cdot \theta|
\left(1 + i\,\frac{2}{\pi}\operatorname{sign}(\xi\cdot \theta)\log |\xi\cdot \theta|\right)
\,\Lambda(d\theta).
\end{equation}
The Caputo time-fractional derivative of a function $u(x,t)$ is defined as
\begin{equation}
\partial_t^{\beta} u(x,t)
:=
\frac{1}{\Gamma(1-\beta)}
\int_{0}^{t}
\frac{\partial_s u(x,s)}{(t-s)^{\beta}}\,ds,
\end{equation}
where $0<\beta<1$ and $\Gamma$ denotes the gamma function.

Let $\mathds{1}_{\gC}(x)$ denote the indicator function of $x$ being inside $\gC$. The main result is as follows. 

\begin{theorem}
\label{thm:safe_prob_pde}
Consider system~\eqref{eq:x_trajectory} and the fractional time change~\eqref{eq:time_change} with the initial state $Y_0 = x \in \gC$. Then, the long-term safety probability~\eqref{eq:safety_probability} is the solution to
\begin{equation}
\label{eq:safety_pde}
\begin{cases}
\partial_T^{\beta}\spy(x,T)=\mathcal{L}\spy(x,T), & x\in \operatorname{int}(\gC),\ T>0,\\
\spy(x,T)=0, & x\in \gC^{c}\cup \partial \gC ,\ T>0,\\
\spy(x,0)=\mathds{1}_{\gC}(x), & x\in \R^n .
\end{cases}
\end{equation}
\end{theorem}

\begin{theorem}
\label{thm:rec_prob_pde}
Consider system~\eqref{eq:x_trajectory} and the fractional time change~\eqref{eq:time_change} with the initial state $Y_0 = x \notin \gC$. Then, the long-term recovery probability~\eqref{eq:recovery_probability} is the solution to
\begin{equation}
\label{eq:recovery_pde}
\begin{cases}
\partial_T^{\beta}\rpy(x,T)=\mathcal{L}\rpy(x,T), & x\in \gC^c,\ T>0,\\
\rpy(x,T)=1, & x\in \gC ,\ T>0,\\
\rpy(x,0)=\mathds{1}_{\gC}(x), & x\in \R^n.
\end{cases}
\end{equation}
\end{theorem}

Theorem~\ref{thm:safe_prob_pde} and Theorem~\ref{thm:rec_prob_pde} show that the long-term safety and recovery probability of the stochastic system~\eqref{eq:time_change} with jumps and memory are solutions of deterministic fractional PDEs. This characterization enables efficient estimation of such long-term risk values for diverse initial states, time horizons and system dynamics, as demonstrated in Section~\ref{sec:experiments}.

\section{Proofs}

In this section, we present the proofs of Theorems~\ref{thm:safe_prob_pde} and~\ref{thm:rec_prob_pde}. We begin by deriving a space-fractional PDE for the safety probability without time-change effects in Lemma~\ref{lm:safe_prob_space_pde}. We then incorporate the time-change mechanism and show that it induces a time-fractional structure in the governing equation, which leads to the proofs of Theorems~\ref{thm:safe_prob_pde} and~\ref{thm:rec_prob_pde}.

\begin{lemma}
\label{lm:safe_prob_space_pde}
Consider system~\eqref{eq:x_trajectory} with the initial state $X_0 = x$. Then, the long-term safety probability $\spx(x,T)
:=
\mathbb{P}\!\left(
X_t \in \gC,\ \forall t \in [0,T] \mid X_0 = x
\right)$ is the solution to
\begin{equation}
\label{eq:space_fractional_pde_x}
\begin{cases}
\partial_T \spx(x,T)=\mathcal{L}\spx(x,T), & x\in \operatorname{int}(\gC),\ T>0,\\
\spx(x,T)=0, & x\in \gC^{c}\cup\partial \gC,\ T>0,\\
\spx(x,0)=\mathds{1}_{\gC}(x), & x\in \R^n.\\
\end{cases}  
\end{equation}
\end{lemma}

\begin{proof}[Proof (Lemma~\ref{lm:safe_prob_space_pde})]
Let $\tau^X_{\gC}:=\inf\{t\ge0: X_t\notin \gC\}$ denote the first exit time of $X$ from $\gC$.
By definition,
\begin{equation}
\spx(x,T)=\mathbb P(\tau^X_{\gC}\geq T \mid X_0 = x).
\end{equation}
For any $0<h<T$, using the Markov property of $X_t$,
\begin{equation}
\spx(x,T)
=
\mathbb E
\left[
\mathds{1}_{\{\tau^X_{\gC}>h\}}
\spx(X_h,T-h) \mid X_0 = x
\right],
\label{eq:dp_safe_prob}
\end{equation}
where $\mathds{1}_{\{\tau^X_{\gC}>h\}}$ is an indicator function of the event $\tau^X_{\gC}>h$. Define $\Phi(s,x):=\spx(x,T-s)$ for $0\le s\le h$. Applying the time-dependent Dynkin's formula~\cite{oksendal2007applied} to the stopped process $X_{s\wedge\tau^X_{\gC}}$ yields
\begin{equation}
\label{eq:dynkin_identity}
\begin{aligned}
    \mathbb E\! & \left[  \Phi(h\wedge\tau^X_{\gC},X_{h\wedge\tau^X_{\gC}}) \mid X_0 = x\right]  - \Phi(0,x) = \\
& \mathbb E \left[\int_0^{h\wedge\tau^X_{\gC}}\left(\partial_s \Phi(s,X_s)+\mathcal L\Phi(s,X_s)\right)ds \mid X_0 = x\right],
\end{aligned}
\end{equation}
where $\wedge$ is the minimum operator such that $h\wedge\tau^X_{\gC} := \min\{h, \tau^X_{\gC}\}$. Since $\spx(x,\cdot)=0$ for $x\in \gC^c\cup\partial\gC$, we have
\begin{equation}
\Phi(h\wedge\tau^X_{\gC},X_{h\wedge\tau^X_{\gC}})
=
\mathds{1}_{\{\tau^X_{\gC}>h\}}
\spx(X_h,T-h).
\end{equation}
From~\eqref{eq:dp_safe_prob}, we know that the left-hand side of~\eqref{eq:dynkin_identity} equals zero, which implies
\begin{equation}
\label{eq:safe_pde_integral}
\mathbb E\left[\int_0^{h\wedge\tau^X_{\gC}}(-\partial_T \spx+\mathcal L \spx )(X_s,T-s)ds \mid X_0 = x\right]=0.
\end{equation}
Dividing~\eqref{eq:safe_pde_integral} by $h$ and letting $h\rightarrow 0$ yields
\begin{equation}
\partial_T \spx(x,T)=\mathcal L \spx(x,T), \qquad x\in\operatorname{int}(\gC), T>0.
\end{equation}
The boundary condition and the initial condition in~\eqref{eq:space_fractional_pde_x} follow from Assumption~\ref{asp:regular_safe_set} and the definitions of the exit time and the long-term safety probability. 
\end{proof}

\begin{lemma}
(Independence Lemma~\cite[Lemma 2.3.4]{shreve2004stochastic})\label{lm:independence}
Let $(\Omega,\mathcal{F},\mathbb{P})$ be a probability space and let 
$\mathcal{G}\subseteq\mathcal{F}$ be a sub-$\sigma$-algebra. Let $Z,W$ be random variables such that $Z$ is $\mathcal{G}$-measurable and $W$ is independent of $\mathcal{G}$. Then for any Borel measurable function $f:\mathbb{R}^m\times\mathbb{R}^n\to\mathbb{R}$ such that $f(Z,W)$ is integrable,
\begin{equation}
\mathbb{E}\!\left[f(Z,W)\mid \mathcal{G}\right]=g(Z)\quad\text{a.s.},
\end{equation}
where $g(z):=\mathbb{E}\!\left[f(z,W)\right]$.
\end{lemma}





\begin{assumption}
\label{asp:continuity_boundedness}
For every $T>0$, the long-term safety probability $\spx(x,T):=\mathbb{P}\!\left(X_t \in \gC,\ \forall t \in [0,T] \mid X_0 = x\right)$ of process~\eqref{eq:x_trajectory} is continuous in $x$, so that $\spx(\cdot,T)\in C_0(\mathbb R^n)$. Furthermore, for almost every $T>0$, $\spx$ is in the domain of the infinitesimal generator $\mathcal L$, i.e., $\spx(x,T) \in D(\mathcal{L})$, and for any $\lambda > 0$, 
\begin{equation}
    \int_0^\infty e^{-\lambda T}\,\|\mathcal L \spx(\cdot,T)\|_{C_0(\mathbb R^n)}\,dT<\infty.
\end{equation}
\end{assumption}

Assumption~\ref{asp:continuity_boundedness} imposes basic continuity and integrability conditions on the long-term safety probability associated with the base process~\eqref{eq:x_trajectory}. These conditions ensure that the Laplace transforms of both $\spx$ and $\mathcal{L}\spx$ are well defined.

\subsection{Proof of Theorem~\ref{thm:safe_prob_pde}}

\begin{proof}[Proof (Theorem~\ref{thm:safe_prob_pde})]
By definition, $\tau_{\gC}\geq T$ if and only if \begin{equation}
\label{eq:exit_time_equiv_Y}
    Y_s\in \gC, \;  \forall s \le T.
\end{equation}
Since $Y_s=X_{E_s}$,~\eqref{eq:exit_time_equiv_Y} is equivalent to
\begin{equation}
\label{eq:exit_time_equiv_X}
X_{E_s}\in \gC, \; \forall s\le T.
\end{equation}
Since the $\beta$-stable subordinator $(D_u)_{u \ge 0}$ is strictly increasing~\cite{meerschaert2019inverse}, $E_T$ is continuous and monotone, and we have 
\begin{equation}
    \{E_s:0\le s\le T\}=[0,E_T].
\end{equation} 
Thus,~\eqref{eq:exit_time_equiv_X} is equivalent to
$
X_r\in \gC, \forall r\le E_T$,
which is also equivalent to
$
\tau^X_{\gC}\geq E_T$.
Therefore,
\begin{equation}
\{\tau_{\gC}\geq T\}=\{\tau^X_{\gC}\geq E_T\}.
\end{equation}
Hence,
\begin{equation}
\spy(x,T)=\mathbb{P}(\tau^X_{\gC}\geq E_T \mid X_0 = x).
\end{equation}
Since $E_T$ is independent of $X$, and hence, independent of $\tau^X_{\gC}$, by Lemma~\ref{lm:independence},
\begin{equation}
\label{eq:repre_combine}
\begin{aligned}
\spy(x,T) = & \; \mathbb{P}(\tau^X_{\gC}\geq E_T\mid X_0 = x) \\ 
= & \; \mathbb{E}\left[\mathbb{P}(\tau^X_{\gC}\geq E_T\mid E_T, X_0 = x)\right] \\
= & \;\mathbb{E}\left[\spx(x,E_T)\right].
\end{aligned}
\end{equation}
From Lemma~\ref{lm:safe_prob_space_pde}, we have $\spx(x,s)$ satisfies the PDE
\begin{equation}
\label{eq:safe_prob_pde_space}
\partial_s \spx(x,s)=\mathcal{L}Q(x,s),
\end{equation}
for $x\in \operatorname{int}(\gC)$ and $s>0$, with initial and boundary conditions
\begin{equation}
\begin{aligned}
\spx(x,0)=\mathds{1}_{\gC}(x), & \qquad x\in \R^n, \\
\spx(x,s)=0, & \qquad x\in \gC^{c}\cup\partial \gC,\ s>0.
\end{aligned}
\end{equation}
From~\cite{meerschaert2013inverse}, the Laplace transform of the density $h(s,T)$ of $E_T$ in $T$ satisfies
\begin{equation}
\tilde h(s,p) := \int_0^\infty e^{-pT}h(s,T)dT=p^{\beta-1}e^{-sp^\beta},\quad \forall p>0.
\end{equation}
Therefore
\begin{equation}
\label{eq:lap_combine}
\int_0^\infty e^{-pT}\mathbb{P}(E_T\in ds)dT=p^{\beta-1}e^{-sp^\beta}ds,\quad \forall p>0.
\end{equation}
From~\eqref{eq:repre_combine},
\begin{equation}
\label{eq:u_v_relationship}
\spy(x,T)=\int_0^\infty \spx(x,s)\mathbb{P}(E_T\in ds).
\end{equation}
Taking the Laplace transform of~\eqref{eq:u_v_relationship} in $T$ gives
\begin{equation}
\begin{aligned}
    \tilde \spy(x,p)& :=\int_{0}^{\infty} e^{-pT}\spy(x,T)\,dT \\
    & =\int_0^\infty \spx(x,s)\left(\int_0^\infty e^{-pT}\mathbb{P}(E_T\in ds)dT\right).
\end{aligned}
\end{equation}
Substituting~\eqref{eq:lap_combine} yields
\begin{equation}
\tilde \spy(x,p)=\int_0^\infty \spx(x,s)p^{\beta-1}e^{-sp^\beta}ds=p^{\beta-1}\tilde \spx(x,p^\beta).
\label{tildeuv}
\end{equation}
Taking the Laplace transform of the PDE~\eqref{eq:safe_prob_pde_space} for $\spx$ with respect to $s$ gives
\begin{equation}
\int_0^\infty e^{-\lambda s}\partial_s \spx(x,s)ds=\lambda\tilde \spx(x,\lambda)-\spx(x,0),\quad \forall \lambda>0.
\end{equation}
Under Assumption~\ref{asp:dynamics_regularity} and~\ref{asp:continuity_boundedness}, the associated semigroup to \eqref{eq:x_trajectory} is Feller~\cite{applebaum2009levy} with closed generator $\mathcal{L}$~\cite{pazy2012semigroups}. Therefore, $\mathcal L$ commutes with the Laplace transform~\cite[Theorem 1.19]{van2008stochastic}, yielding
\begin{equation}
\int_0^\infty e^{-\lambda s}\mathcal L \spx(x,s)ds=\mathcal L\tilde \spx(x,\lambda),\quad \forall \lambda>0.
\end{equation}
Therefore
\begin{equation}
\lambda\tilde \spx(x,\lambda)-\spx(x,0)=\mathcal L\tilde \spx(x,\lambda),\quad \forall \lambda>0.
\end{equation}
Since $\spx(x,0)=1$ when $x\in \gC$,
\begin{equation}
(\lambda-\mathcal L)\tilde \spx(x,\lambda)=1,\quad \forall \lambda>0, x\in \gC.
\label{trans_combine}
\end{equation}
Setting $\lambda=p^\beta$ and using \eqref{tildeuv} gives
\begin{equation}
(p^\beta-\mathcal L)\frac{\tilde \spy(x,p)}{p^{\beta-1}}=1,\quad \forall p>0,
\end{equation}
which implies
\begin{equation}
p^\beta\tilde \spy(x,p)-p^{\beta-1}=\mathcal L\tilde \spy(x,p),\quad \forall p>0.
\label{gen_lap}
\end{equation}
The Laplace transform of the Caputo derivative satisfies
\begin{equation}
\widetilde{\{\partial_T^\beta \spy\}}(p)=p^\beta\tilde \spy(x,p)-p^{\beta-1}\spy(x,0),\quad p>0.
\end{equation}
Since $\spy(x,0)=\mathbb{P}(\tau_{\gC}\geq 0 \mid Y_0 = x \in \gC)=1$,
\begin{equation}
\widetilde{\{\partial_T^\beta \spy\}}(p)=p^\beta\tilde \spy(x,p)-p^{\beta-1},\quad p>0.
\end{equation}
Comparing with \eqref{gen_lap} yields
\begin{equation}
\widetilde{\{\partial_T^\beta \spy\}}(p)=\mathcal L\tilde \spy(x,p)=\widetilde{\mathcal{L}F}(x,p),\quad \forall p>0.
\end{equation}
The uniqueness of Laplace transform implies 
\begin{equation}
\partial_T^\beta \spy(x,T)=\mathcal L \spy(x,T),\qquad x\in \operatorname{int}(\gC),\ T>0.
\end{equation}
The initial and boundary conditions remain identical to those established in Lemma \ref{lm:safe_prob_space_pde}. This spatial invariance holds due to the nature of the subordinated process $Y_t = X_{E_t}$. Since the inverse stable subordinator $E_t$ is a continuous and non-decreasing operational time process with $E_0 = 0$, the starting position is strictly preserved ($Y_0 = X_0$). Furthermore, because the time-change purely stretches or compresses time without altering the underlying geometric path of the base process, any point $x \in \partial \gC$ that is regular for the base process $X_t$ remains regular for the time-changed process $Y_t$. Thus, the process satisfies the same zero-boundary conditions upon exiting the safe region.
\end{proof}

\subsection{Proof of Theorem~\ref{thm:rec_prob_pde}}

\begin{proof}[Proof (Theorem~\ref{thm:rec_prob_pde})]
Since both $\gC^c$ and $\operatorname{int}(\gC)$ are open sets, let us consider a new safe set $\gC^c \cup \partial \gC$ as opposed to $\gC$. By Theorem~\ref{thm:safe_prob_pde}, the corresponding safety probability~\eqref{eq:safety_probability} will satisfy the following PDE
\begin{equation}
\label{eq:safe_prob_pde_reversed}
    \partial_T^{\beta}\spy(x,T)=\mathcal{L}\spy(x,T), \qquad x\in \gC^c,\ T>0.
\end{equation}
Since exiting $\gC^c$ is equivalent to returning to $\gC$, we have
\begin{equation}
\label{eq:safe_rec_prob_relation}
\begin{aligned}
\rpy(x,T) & := \mathbb{P}(\rho_{\gC} \leq T \mid Y_0 = x \notin \gC) \\
    &=\mathbb{P}(\tau_{\gC^c}\le T \mid Y_0 = x \notin \gC) \\
    & =1-\mathbb{P}(\tau_{\gC^c}\geq T \mid Y_0 = x \notin \gC) \\
    &=1-\spy(x,T).
\end{aligned}
\end{equation}
Since~\eqref{eq:safe_prob_pde_reversed} is a linear PDE, plugging in~\eqref{eq:safe_rec_prob_relation} we have
\begin{equation}
\partial_T^\beta \rpy(x,T)=\mathcal L \rpy(x,T),\qquad x\in \gC^c,\ T>0.
\end{equation}
By definition of the recovery probability~\eqref{eq:recovery_probability}, we get
\begin{equation}
\begin{aligned}
    \rpy(x,T)=1, & \qquad x\in \gC,\ T>0,\\
    \rpy(x,0)=\mathds{1}_{\gC}(x), &\qquad x\in \R^n.
\end{aligned}
\end{equation}
\end{proof}

\begin{remark}
    Although we consider safety probability and recovery probability of systems with both jumps and memory, the proof entails that the effect of jumps and memory can be separately analyzed, resulting in space- and time-fractional terms in the PDE characterization, respectively. This is verified in the experiments in Section~\ref{sec:rec_prob_exp}.
\end{remark}

\section{Experiments}
\label{sec:experiments}


This section presents several experimental results demonstrating the effectiveness of the fractional PDE characterization of long-term probabilities for stochastic systems with jumps and memory. We also demonstrate how physics-informed machine learning can be leveraged based on the derived PDEs, for efficient estimation of probability values of diverse systems dynamics.

\subsection{Recovery Probabilities}
\label{sec:rec_prob_exp}


We first study the recovery probability~\eqref{eq:recovery_probability} for stochastic systems with jumps and memory, and demonstrate the effectiveness of the proposed fractional PDE characterization. Specifically, we consider the following case study: 

\begin{itemize}
    \item System~\eqref{eq:x_trajectory} with $f \equiv 0.8$ and $\sigma \equiv 0.4$, with $\alpha = 1$ (with Lévy jumps) and $\beta = 1$ (without fractional time).
    \item System~\eqref{eq:x_trajectory} with $f \equiv 0.3$ and $\sigma \equiv 0.6$, with $\alpha = 2$ (without jumps) and $\beta = 0.4$ (with fractional time).
\end{itemize}
The safe set is defined as $\gC = \{x : x > 1\}$. We obtain reference estimates using Monte Carlo (MC) simulation and solve the corresponding fractional PDEs. 
For comparison, Fig.~\ref{fig:rec_space_pde_vs_mc} and Fig.~\ref{fig:rec_time_pde_vs_mc} visualize the MC estimates and proposed PDE solutions across the state space at different time horizons. 
We also include results for the corresponding systems without jumps or memory, together with their associated classical (non-fractional) PDE formulations. 
The results show that the space- and time-fractional PDEs accurately characterize the long-term risk probabilities for systems with Lévy jumps and memory effects, respectively, across a range of initial states $x$ and time horizons $T$. 
Importantly, these systems exhibit long-term probability values that differ substantially from those of the same underlying dynamics when jumps or memory are absent. 
As a result, classical convection–diffusion PDEs cannot capture these behaviors, highlighting that fractional PDEs are essential for accurately characterizing risk in stochastic systems with jumps and memory.

\begin{figure}[t]
    \centering
    \includegraphics[width=0.99\columnwidth]{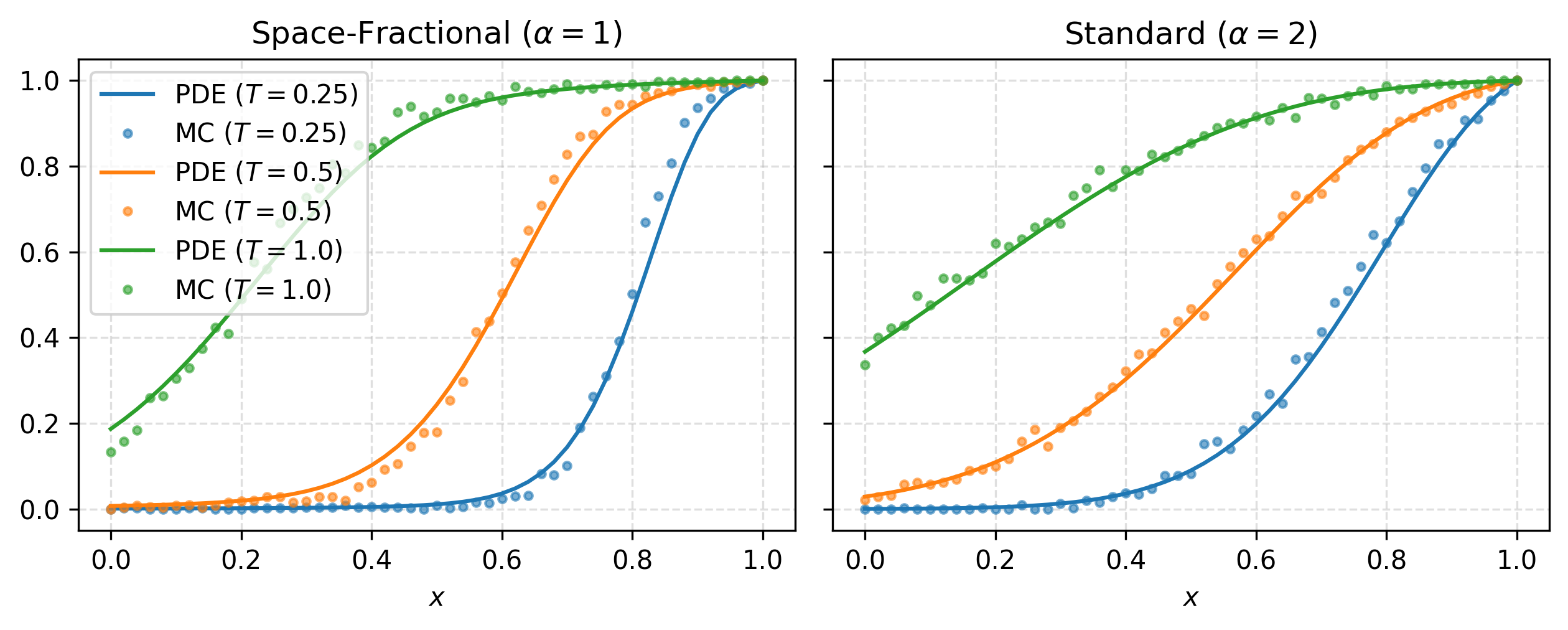}
    \caption{Recovery probability with and without Lévy jumps.
    }
    \vspace{-1.2em}
    \label{fig:rec_space_pde_vs_mc}
\end{figure}

\begin{figure}[t]
    \centering
    \includegraphics[width=0.99\columnwidth]{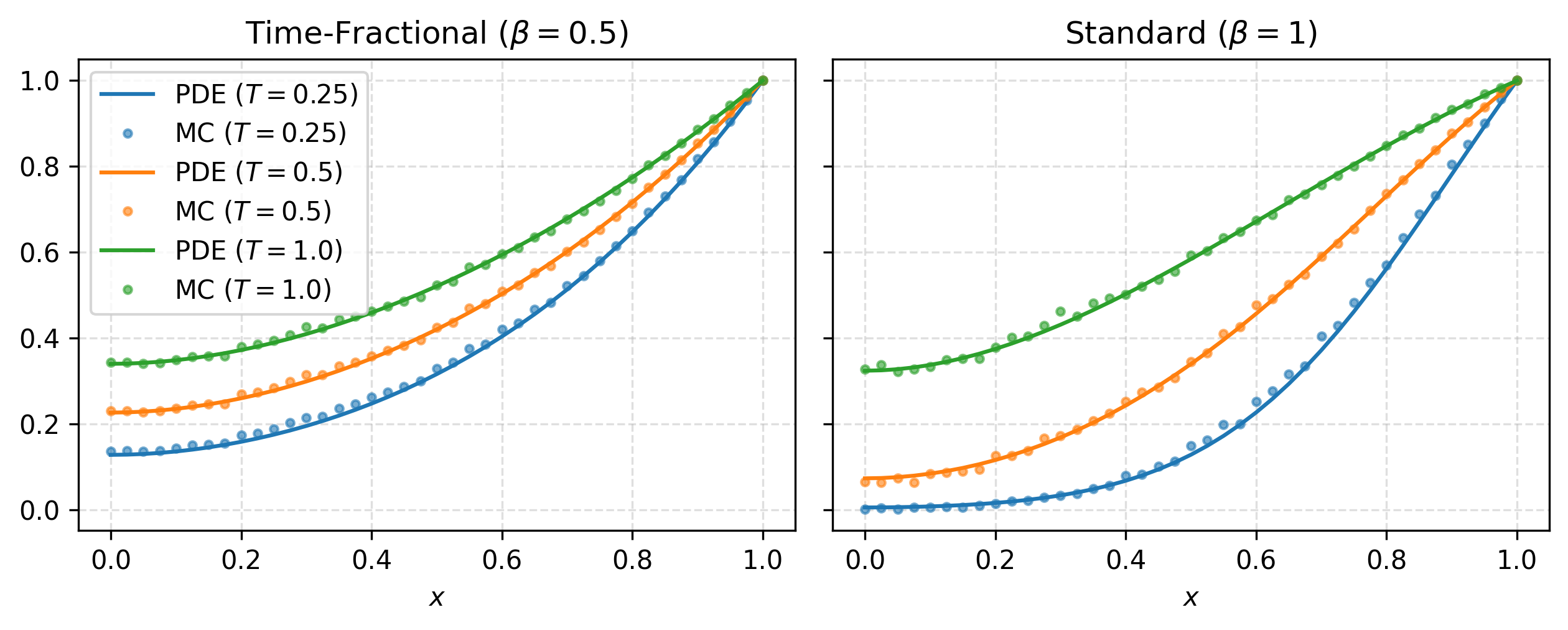}
    \caption{Recovery probability with and without time change.
    }
    \vspace{-1.5em}
    \label{fig:rec_time_pde_vs_mc}
\end{figure}

\subsection{Safety Probabilities}
\label{sec:safe_prob_exp}

Next, we show that the derived PDE characterization enables physics-informed learning methods to efficiently estimate long-term safety probabilities~\eqref{eq:safety_probability} across diverse systems. Specifically, we consider the safety probability for a diverse class of systems~\eqref{eq:x_trajectory} and~\eqref{eq:time_change} in two dimensions with Lévy jumps ($\alpha=1.5$) and fractional time change ($\beta=0.7$). We train a physics-informed Fourier neural operator (PINO)~\cite{li2024physics} to learn the corresponding time-space fractional PDE~\eqref{eq:safety_pde} to estimate the safety probabilities.

To evaluate generalization, the system dynamics are sampled from a randomized family. Specifically, for the system dynamics function $f := [f_1, f_2]$ corresponding to the two spatial coordinates $x_1$ and $x_2$, both $f_1$ and $f_2$ are independently drawn from
\begin{align}
\label{eq:training_family}
f(x_1,x_2)
&=
c_0
+
c_{1,1}x_1
+
c_{1,2}x_2
+
c_{2,1}x_1^2
+
c_{2,2}x_2^2 \notag\\
&+
c_{12}x_1x_2
+
c_{3,1}|x_1|^{1.5}
+
c_{3,2}|x_2|^{1.5} \notag\\
&+
\sum_{k_1=1}^{K_1}\sum_{k_2=1}^{K_2}
\left(
\frac{
a_{k_1,k_2}\sin\bigl(2\pi(s_1k_1x_1+s_2k_2x_2)\bigr)
}{
(k_1^2+k_2^2)^{d/2}
}
\right. \notag\\
&\left.
+
\frac{
b_{k_1,k_2}\cos\bigl(2\pi(s_1k_1x_1+s_2k_2x_2)\bigr)
}{
(k_1^2+k_2^2)^{d/2}
}
\right),
\end{align}
where the coefficients are independently sampled from prescribed bounded domains. We consider the spatial domain $(x_1, x_2) \in [0,1]^2$ and the time horizon $T \in [0,1]$, with safe set $\gC = \{x_1 < 1, x_2 < 1\}$. We generate ground-truth solutions over the entire space-time domain for 1000 randomly sampled dynamics $f$, among which 800 are used for training and 200 for testing.

The model is trained using the physics-informed loss
\begin{equation}
\mathcal{L}_{\mathrm{loss}}
=
\mathcal{L}_{\mathrm{data}}
+
\lambda_{\mathrm{PDE}}\mathcal{L}_{\mathrm{PDE}}
+
\lambda_{\mathrm{BC}}\mathcal{L}_{\mathrm{BC}},
\end{equation}
where $\mathcal{L}_{\mathrm{data}}=\|\hat{\spy}-\spy_{\mathrm{true}}\|_2^2$ is the mean-squared error loss between the prediction $\hat{\spy}$ and the ground truth $\spy_{\mathrm{true}}$, $\mathcal{L}_{\mathrm{PDE}}$ is the fractional PDE residual loss with weight $\lambda_{\mathrm{PDE}}$, and $\mathcal{L}_{\mathrm{BC}}$ is the boundary conditions residual loss with weight $\lambda_{\mathrm{BC}}$. After 100 epochs of training, the in-distribution test relative $L^2$ error is $9.77\times10^{-3}$, indicating that the learned operator accurately predicts safety probabilities for unseen systems within this diverse family of system dynamics.

We further evaluate the model on out-of-distribution (OOD) dynamics of the form
\begin{align}
\label{eq:test_ood}
f_1(x_1,x_2)
&=1.32x_1+1.32x_2+1, \notag\\
f_2(x_1,x_2)
&=10|x_1|^{1.6}
-5|x_2|^{1.6}
+2\cos(3\pi x_2)
+1 \notag\\
&+20\sin(0.96x_1+0.36x_2+\pi).
\end{align}
Fig.~\ref{fig:2D_ood_results} presents the ground-truth safety probability, the PINO predicted solution, and the corresponding error at different time instances. The results show that the model can produce accurate predictions even for OOD dynamics (relative $L^2$ error $2.37\times10^{-2}$), which are fundamentally challenging for existing sampling-based methods (Table~\ref{tab:pino_mc_comparison}).

\begin{table}[t]
    \centering
    \footnotesize
    \setlength{\tabcolsep}{5pt}
    \renewcommand{\arraystretch}{1.0}
    \caption{OOD prediction time and error comparison.}
    \label{tab:pino_mc_comparison}
    \vspace{-0.5em}
    \begin{tabular}{lcc}
        \toprule
         & MC & PINO \\
        \midrule
        Time (s) / Rel. $L^2$ Error & 93.42 / $9.06{\times}10^{-2}$ & \textbf{0.26} / $\mathbf{2.37{\times}10^{-2}}$ \\
        \bottomrule
    \end{tabular}
    \vspace{-0.5em}
\end{table}

\begin{figure}[t]
 \centering
 \includegraphics[width=0.49\textwidth]{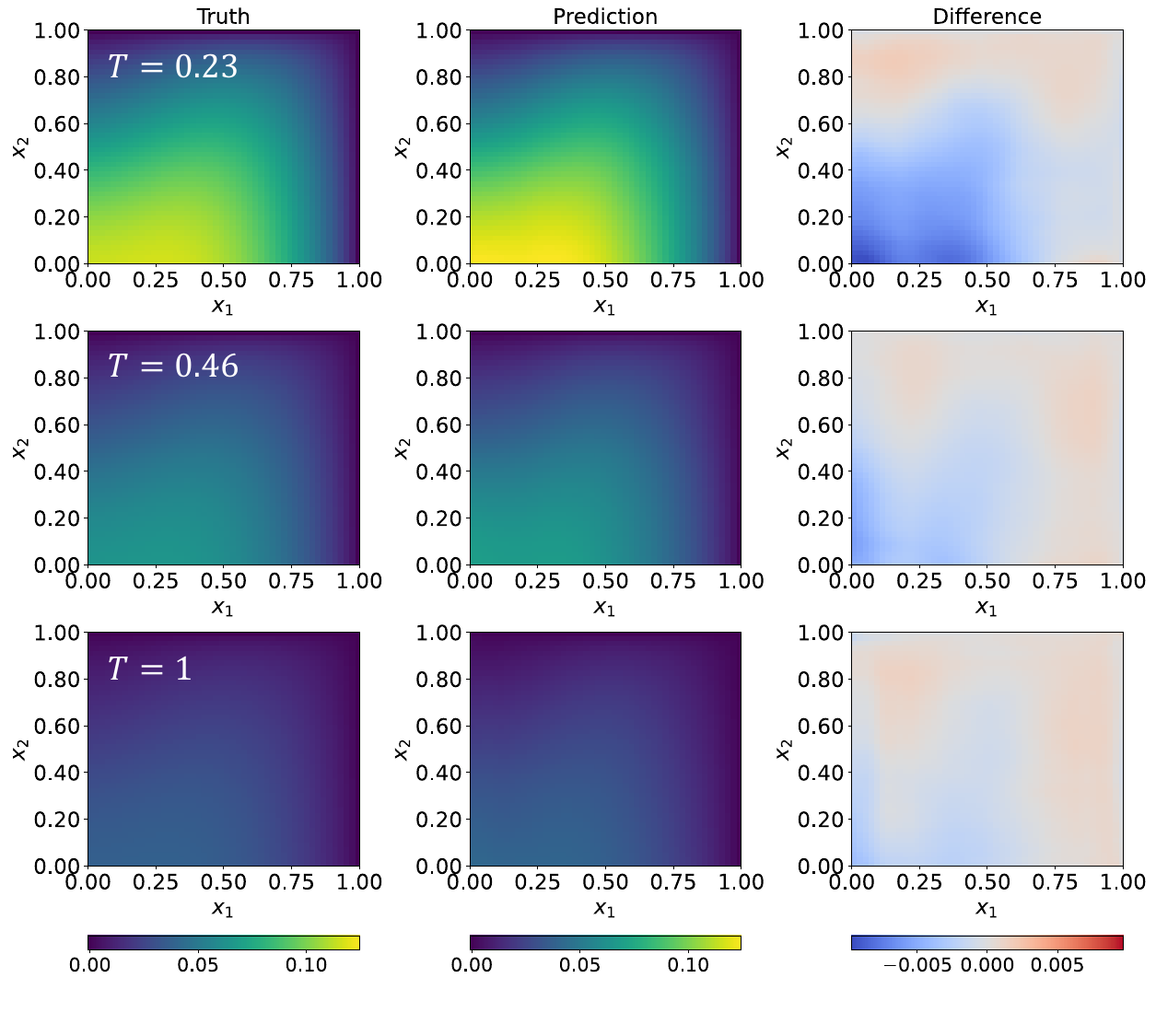}
 \vspace{-1.5em}
 \caption{Safety probability for $2$D system with OOD dynamics. 
 }
 \vspace{-1em}
 \label{fig:2D_ood_results}
\end{figure}

\section{Conclusions}

In this paper, we studied long-term safety and recovery probabilities for stochastic dynamical systems with jumps and memory effects. We showed that these probabilities admit an exact characterization through fractional PDEs, extending existing risk analysis beyond Gaussian disturbances and incorporating fractional time-change dynamics. This formulation provides a principled framework that transforms stochastic safety evaluation into the solution of deterministic fractional PDEs. Building on this connection, we demonstrated that physics-informed learning can efficiently approximate these solutions and enable risk quantification across diverse initial states and system configurations. Future work includes developing improved computational methods for high-dimensional systems.

\section*{Appendix}



We describe the experimental setup. For recovery probability (Section~\ref{sec:rec_prob_exp}), the space-fractional case uses a generator-based discretization with implicit Euler, while MC estimates are obtained via continuous-time Markov chain (CTMC) simulation. For the time-fractional case, the Caputo derivative is discretized using the $L_1$ scheme, and MC estimates are computed via a subordination-based method.
For safety probability (Section~\ref{sec:safe_prob_exp}), ground-truth solutions are obtained by solving the fractional PDE with a sparse CTMC generator and implicit $L_1$ scheme. The neural operator takes input $[f_1,f_2,x_1,x_2,t]$ and outputs the safety probability, trained with $\lambda_{\mathrm{PDE}}=0.01$ and $\lambda_{\mathrm{BC}}=1$. Experiments are run on an RTX 5080 GPU and AMD 9800X3D CPU.


\bibliography{citation.bib}

\end{document}